\documentclass[letterpaper, 10 pt, journal, twoside]{IEEEtran}
\usepackage{url}
\usepackage{graphicx}
\usepackage{wrapfig}
\usepackage[format=plain,font=footnotesize,labelfont=bf,labelsep=period]{caption}
\usepackage{sidecap} 
\usepackage[export]{adjustbox}
\usepackage{subcaption}
\usepackage{stfloats}
\usepackage[font=small]{caption}
\usepackage{float}

\usepackage{amsmath} 
\usepackage{amssymb}  
\usepackage{amsthm}
\usepackage{mathtools}
\usepackage{pdfsync}
\usepackage[normalem]{ulem}
\usepackage{paralist}	
\usepackage[space]{grffile} 
\usepackage{color}

\usepackage{resizegather}

\newtheorem{theorem}{Theorem}
\newtheorem{corollary}{Corollary}
\newtheorem{proposition}{Proposition}
\newtheorem{lemma}{Lemma}
\theoremstyle{definition}
\newtheorem{definition}{Definition}
\theoremstyle{remark}
\newtheorem{remark}{Remark}
\theoremstyle{definition}

\theoremstyle{definition}
\newtheorem{example}{Example}

\usepackage{adjustbox}

\newcommand{\newsec}[1]{
\vspace{0.2cm}
\noindent \textbf{#1}
}

\newcommand{\R}{\mathbb{R}}
\renewcommand{\S}{\mathcal{S}}
\newcommand{\D}{\mathcal{D}}
\newcommand\norm[1]{\left\lVert#1\right\rVert}

\newcommand{\classK}{\mathcal{K}}

\newcommand{\q}{q}
\newcommand{\dq}{\dot{q}}

\newcommand{\alert}[1]{\textcolor{red}{#1}}

\IEEEoverridecommandlockouts
\allowdisplaybreaks

\newcommand\revtwo{}

\begin{document}

\theoremstyle{definition}
\theoremstyle{remark}
\theoremstyle{definition}
\theoremstyle{definition}

\title{Safety-Critical Kinematic Control of \\ Robotic Systems}
\author{Andrew Singletary$^1$, Shishir Kolathaya$^2$, and Aaron D. Ames$^1$
\thanks{$^1$Mechanical and Civil Engineering, Caltech, Pasadena, CA, USA, 91125.
        {\tt\small \{ames,asinglet\}@caltech.edu}}
\thanks{$^2$Cyber Physical Systems and Computer Science and Automation, IISc, Bengaluru, KA, India, 560012.
        {\tt\small \{shishirk\}@iisc.ac.in}}
}
\maketitle
\thispagestyle{empty}

\begin{abstract}

Over the decades, kinematic controllers have proven to be practically useful for applications like set-point and trajectory tracking in robotic systems. To this end, we formulate a novel safety-critical paradigm for kinematic control in this paper. In particular, we extend the methodology of control barrier functions (CBFs) to kinematic equations governing robotic systems. We demonstrate a purely kinematic implementation of a velocity-based CBF, and subsequently introduce a formulation that guarantees safety at the level of dynamics. This is achieved through a new form CBFs that incorporate kinetic energy with the classical forms, thereby minimizing model dependence and conservativeness. The approach is then extended to underactuated systems. This method and the purely kinematic implementation are demonstrated in simulation on two robotic platforms: a 6-DOF robotic manipulator, and a cart-pole system.
\end{abstract}


\section{Introduction}
\IEEEPARstart{K}{inematic} control provides a powerful method for achieving desired behaviors on a large class of robotic systems \cite{siciliano1990kinematic,xiang2010general,antonelli2006kinematic}. In a variety of applications, such as industrial manipulators and commercial drones, the low-level torque controllers are pre-configured, and end-users are only able to interface with the system through desired velocity inputs. Also, the dynamical model of the system is not well known by the users, making it difficult to provide guarantees on safety. Hence, these observations point to an important limitation in existing work on safeguarding control laws: there is a very high dependence on model, and there is a lack of general framework for automatically incorporating the velocity based low-level tracking with the high level control.

Ensuring safety via quadratic programming based control laws were popularized by \cite{ames17}, where safety constraints were incorporated via control barrier functions (CBFs). This was first applied to adaptive cruise control, and has since been utilized in a variety of application domains: automotive safety \cite{jankovic2018robust}, robotics \cite{cortez2020correct,cortez2019control,rauscher2016constrained} and multi-agent systems \cite{Wang17,lindemann2019control}. See \cite{ames2019control} for a recent survey. While CBFs can be implemented in a purely kinematic fashion for robotic systems \cite{landi2019safety}, as will be demonstrated in this work, it only guarantees safety kinematically, not the full robotic dynamics. 
As mentioned previously, the constraints are heavily dependent on the 
model, and safety cannot be guaranteed with model uncertainty.


Recently, energy-based reciprocal control barrier functions were introduced \cite{KolathayaEnergyCBF2020} as a means to provide robust safety guarantees for fully-actuated robotic platforms with model uncertainty. This was done by utilizing bounds on the inertia and Coriolis-centrifugal matrices, as well as the gravity vector, and providing safety guarantees for the worst-case scenario. While the resulting QP formulation yielded robustness in safety, it does not have well-defined behavior on the boundary of the set and outside of it, making it difficult to implement in practice. Therefore, in this paper, an alternative formulation for the energy-based CBFs is introduced for zeroing control barrier functions, which are well defined on the boundary, and exterior of the set. Additionally, in this formulation, the added conservatism inherent in reducing model dependence is minimized and analyzed. This method is then utilized to provide guarantees for velocity control inputs, using partial model information.
This analysis is then extended to the class of underactuated robotic systems. 
The results are demonstrated in a $6$-DOF manipulator and a cart-pole system (see Figs. \ref{fig:intro} and \ref{fig:cart_pole}), wherein different levels of uncertainties are incorporated and safety-critical kinematic control laws are applied.

\begin{figure}[t]
\centering
  \centering
  \includegraphics[width=.75\columnwidth]{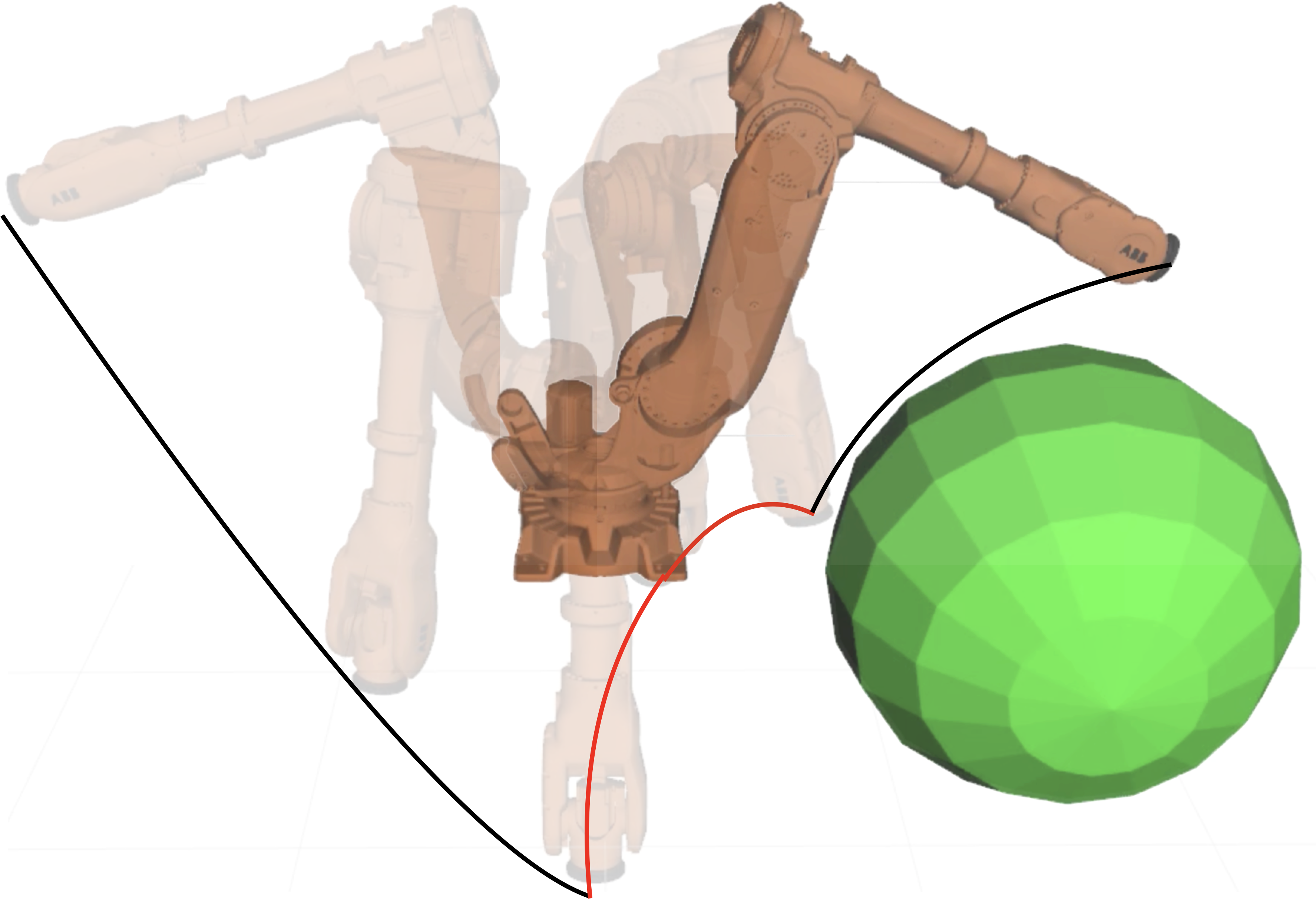}  
    \caption{A $6$-DOF manipulator safely avoiding an obstacle with energy-based control barrier function. The CBF intervention is shown in red.
    A video can be found at \cite{youtubeclip}. 
    }
    \label{fig:intro}
    \vspace{-0.45cm}
\end{figure}


This paper is structured as follows. Section \ref{sec:background} provides the necessary background on CBFs. Section \ref{sec:kinematic} demonstrates safety-critical velocity control of purely kinematic systems, with no regard for the underlying dynamics of the system. In Section \ref{sec:dynamics}, we begin with the formulation of an energy-based CBF that guarantees the safety of a robotic system at the dynamics level, using minimal model information. Then, this formulation is modified to guarantee safety of the dynamical system for kinematic control inputs, in this case a desired velocity command. The results are demonstrated in simulation on a 6 DOF robotic manipulator, and a comparison is made to the purely kinematic case. Finally, in Section \ref{sec:underactuated}, the underactuated case is explored, and the method is demonstrated with a simulation of a cart-pole system.

\section{Background}
\label{sec:background}
Consider a nonlinear control system in affine form: 
\begin{equation}
\label{eqn:controlsys}
    \dot{x} = f(x) + g(x) u
\end{equation}
where $x \in \D \subset \R^n$ is the state, and $u \in U \subseteq \R^m$ the input. 
Assume that the functions $f:\R^n\rightarrow \R^n$ and $g:\R^n\rightarrow \R^{n \times m}$ are continuously differentiable (therefore, for a Lipschitz continuous controller solutions exist and are unique).   We are interested in \emph{safety} defined as the forward invariance of a set $\S \subset \D$.  In particular, given a feedback control law $u = k(x)$ the resulting closed loop system $\dot{x} = f_{\rm cl}(x) = f(x) + g(x) k(x)$ with solution 
$x(t)$ and initial condition $x(0) = x_0$.  The system is \emph{safe} with respect to the controller $u = k(x)$ if: 
$$
\forall ~ x_0 \in \S \qquad \Rightarrow \qquad 
x(t) \in \S ~ \forall ~ t \geq 0. 
$$

\begin{definition}[\cite{ames17}]
\label{def:cbf}
Let $\S \subset \D \subset \R^n$ be the \underline{0-superlevel set} of a continuously differentiable function $h: \D \to \R$:
\begin{eqnarray}
\S & = & \{ x \in \R^n ~ : ~ h(x) \geq 0 \} , \nonumber\\
\partial \S & = & \{ x \in \R^n ~ : ~ h(x) = 0 \}, \nonumber\\
\mathrm{Int}(\S) & = & \{ x \in \R^n ~ : ~ h(x) > 0 \}. \nonumber
\end{eqnarray}
Then $h$ is a \underline{control barrier function (CBF)} if $\frac{\partial h}{\partial x}(x)\neq 0$ for all $x\in\partial \S$ and there exists an \emph{extended class $\classK$ function} (\cite[Definition 2]{ames17})
$\alpha$ such that for the control system \eqref{eqn:controlsys} and for all 
$ x \in \S$:
\begin{align}
\label{eqn:cbf:definition}
 \sup_{u \in U}  \Big[ \underbrace{L_f h(x) + L_g h(x) u}_{\dot{h}(x,u)} \Big] \geq - \alpha(h(x)).
\end{align}
where $L_f h(x) = \frac{\partial h}{\partial x} f(x)$ and $L_g h(x) = \frac{\partial h}{\partial x} g(x)$.  We say that $h$ is a \underline{control barrier function (CBF) on $\S$} if \eqref{eqn:cbf:definition} holds for all $x \in \S$ (but not necessarily on all of $D$). 
\end{definition}


The main result with regard to control barrier functions, is that the existence of a control barrier function implies that the control system is safe:

\begin{theorem}[\cite{ames17}]
\label{thm:cbf}
{\it Let $\S \subset \R^n$ be a set defined as the superlevel set of a continuously differentiable function $h: \D \subset \R^n \to \R$.  If $h$ is a control barrier function (CBF) on $\S$, 
then any Lipschitz continuous controller satisfying: $\dot{h}(x,u(x)) = L_f h(x) + L_g h(x) u(x) \geq - \alpha(h(x)) $ renders the set $\S$ safe for the system \eqref{eqn:controlsys}.  Additionally, if $h$ is a CBF on $\D$, the set $\S$ is asymptotically stable in $\D$. }
\end{theorem}


\newsec{Controller Synthesis.}  The main idea underlying controller synthesis with barrier functions is to use them as \emph{safety filters} which take in a desired control input, $u_{\rm des}(x,t)$, and modifies this input in a minimal way so as to guarantee safety.  This can be formalized as an optimization based controller, and specifically a Quadratic Program (QP).  Specifically: 
\begin{align}
\label{eq:cbfqp}
\tag{CBF-QP}
u^*(x) = \underset{u \in U \subseteq \R^{m}}{\operatorname{argmin}} &  \quad  \| u - u_{\rm des}(x,t) \|^2  \\
\mathrm{s.t.} &  \quad L_f h(x) + L_g h(x) u \geq - \alpha(h(x))  \nonumber
\end{align}
Importantly, this controller has an explicit solution as noted by the following lemma.

\begin{lemma}
\label{lem:explicitCBFsol}
Let $h$ be a control barrier function for the control system \eqref{eqn:controlsys} and assume that $U = \R^m$.  Then the explicit solution to the QP \eqref{eq:cbfqp} is given by: 
\begin{eqnarray}
\label{eqn:minnormexplicit}
~  u^*(x,t)  =  u_{\rm des}(x,t) + u_{\rm safe}(x,t) 
\end{eqnarray}
where $u_{\rm safe}$ minimally modifies $u_{\rm des}$ depending on if the nominal controller keeps the system safe, i.e., the sign of $\Psi(x,t ; u_{\rm des}) := \dot{h}(x,u_{\rm des}(x,t))  + \alpha(h(x))$, according to:
\begin{gather}\label{eq:firstsolutiontoqp}
u_{\rm safe}(x,t) =  \left\{ 
\begin{array}{lcr}
- \frac{L_g h(x)^{T}}{L_g h(x) L_g h(x)^T} \Psi(x,t ;u_{\rm des}) & \mathrm{if~} \Psi(x,t ;u_{\rm des}) < 0 \\
0 & \mathrm{if~} \Psi(x,t ;u_{\rm des}) \geq 0
\end{array}
\right.  
\end{gather}
\end{lemma}

\begin{proof}
In \cite{xu2015robustness}, an explicit form for \eqref{eq:cbfqp} was found using the KKT conditions when $u_{\rm des}(x,t) = 0$.  The same proof with a modified cost yields the desired result.  Specifically, the dual-primal feasibility and complementary slackness conditions remain unchanged.  The stationary condition becomes: 
$
u^*(x,t) = \mu(x) L_g h(x)^T  + u_{\rm des}(x,t). 
$
Utilizing this in the proof results in the closed form solution.  Finally, safety is guaranteed from Theorem \ref{thm:cbf}.
\end{proof}

\section{Safety-Critical Kinematic Control}
\label{sec:kinematic}
In this section, we consider safety-critical kinematic control. We provide an example of velocity-based kinematic control of a robotic manipulator, and analyze its ability to maintain safety of the dynamical system.

In the context of kinematic control of robotic systems, we are interested in kinematic mappings of the form: 
$x = y(q)$ where $q \in Q \subset \R^k$, $x \in \D \subset \R^n$ and thus $y : Q \to \D$.  Here, we assume that $k \geq n$, i.e., that there are more degrees of freedom than tasks. 
Here $x$ is the vector of ``outputs'' or ``task'' variables, i.e., a vector of elements which we wish to control, and $q$ is a vector consisting of the systems configuration, e.g., angles of the robotic system.  The evolution of the task variables is therefore given by: 
\begin{eqnarray}
\label{eqn:firstorderkin}
\dot{x} = J_{y}(q) \dot{q}. 
\end{eqnarray}
In kinematic control, we view $\dot{q}$ as 
the input to the system.  Specifically, we wish to determine a feedback control law:  $\dot{q} = K(q,t)$ that achieves the desired properties.  

\newsec{Kinematic Trajectory Tracking.}  Suppose that we have a desired trajectory $x_d(t)$ for the task vector.  The goal is to track this trajectory, i.e., for $e(t) = x(t) - x_d(t) \to 0$ with $x(t)$ satisfying \eqref{eqn:firstorderkin}. Differentiating this yields: 
$$
\dot{e} = J_y(q) \dot{q} - \dot{x}_d(t)
$$
Therefore, for $\gamma > 0$, if we pick:
$$
J_y(q)\dot{q} = \dot{x}_d(t)  - \gamma e \:\: \Rightarrow \:\:
\dot{e} = - \gamma e 
\:\: \Rightarrow \:\:
e(t) = \exp(-\gamma t) e(0),
$$
and the end result is a stable linear system on the output.  As a result, if we wish to track a trajectory, we can pick: 
\begin{align}
\label{eqn:qdotdestrajectory}
\dot{q}_{\rm des}(x,t) = J_y(q)^{\dagger}\left(\dot{x}_d(t) - \lambda e\right),
\end{align}
with $J_y(q)^{\dagger} = J_y(q)^{T}  (J_y(q) J_y(q)^T)^{-1}$, the Moore-Penrose (right) pseudoinverse, assumed to be well defined.
This results in exponential stability in the error dynamics: $e(t) \leq \exp(-\lambda t) e(0)$.


\newsec{Safety-Critical Control.}  We will now study safety-critical velocity based control. We have the following.


\begin{lemma}
\label{lem:velocity}
Consider a kinematic safety constraint $h : Q \subset \R^k \to \R$ and the corresponding safe set $\S = \{ q \in Q ~ : ~ h(q) \geq 0\}$ defined as the 0-superlevel set of $h$.  If $J_h(q) \neq 0$, then the following velocity based controller:
\begin{align}
\label{eqn:QPvelocity}
\dot{q}^*(q,t) = \underset{\dot{q} \in \R^k}{\operatorname{argmin}} & ~  ~ \|  \dot{q} - J_y(q)^{\dagger} \left(\dot{x}_d(t) - \lambda (y(q) - x_d(t)) \right) \|^2 \nonumber \\
\mathrm{s.t.} & ~  ~ \dot{h}(q,\dot{q}) =  J_h(q) \dot{q} \geq - \alpha(h(q)), 
\end{align}
ensures safety, i.e., $\S$ is forward invariant.  Moreover, this has a closed form solution given by \begin{gather}
\begin{split}
\label{eqn:minnormexplicitsimple}
 ~  \dot{q}^*(x,t)  = \dot{q}_{\rm des}(q,t) + 
\left\{ 
\begin{array}{lcr}
-J_h(x)^{\dagger} \Psi(q,t; \dot{q}_{\rm des})  & \mathrm{if~} \Psi(q,t;\dot{q}_{\rm des}) < 0 \\
0 & \mathrm{if~} \Psi(q,t;\dot{q}_{\rm des})  \geq 0
\end{array}
\right.  
\end{split}
\end{gather}
where $\Psi(x,t;\dot{q}_{\rm des}) = J_h(q) \dot{q}_{\rm des}(q,t)  + \alpha(h(q)) $.
\end{lemma}


Therefore, the controller \eqref{eqn:minnormexplicitsimple} utilizes $\dot{q}_{\rm des}$ whenever it is safe, i.e., when $\Psi(q,t;\dot{q}_{\rm des}) \geq 0$.  Conversely, in the case when $\dot{q}_{\rm des}$ is unsafe the controller takes over and enforces $\dot{h} = J_h(q) \dot{q}^*(q,t) = - \alpha(h)$ until $\dot{q}_{\rm des}$ is safe again.



\begin{example}[\textbf{Manipulator Obstacle Avoidance}]
\label{ex:manipulator_kinematic_only}
Consider a 6-DOF industrial manipulator (see Fig. \ref{fig:intro}) attempting to track a desired trajectory $x_d(t)$ using the desired velocity given in \eqref{eqn:qdotdestrajectory} with its end-effector. 
Note that CBFs have been successfully applied to robot manipulators in \cite{cortez2019control,rauscher2016constrained,landi2019safety} via kinematic control. Suppose that the manipulator needs to complete this trajectory while avoiding an obstacle located at $(x_0,y_0,z_0)$. Thus, in the set $\mathcal{S} = \{q \mid h(q) \geq 0\}$, the end-effector must be at least a distance $d$ from the obstacle.
A control barrier function representing this safety constraint is
\begin{align}
\label{eq:manip_h}
    h(x) = (x-x_0)^2 + (y-y_0)^2 + (z-z_0)^2 - d^2.
\end{align}
Its Jacobian is given by
\begin{align}
    J_h(q) = \frac{dh}{dx} \frac{dx}{dq} = \begin{bmatrix}
    2x-2x_0 &
    2y-2y_0 &
    2z-2z_0
    \end{bmatrix} J,\nonumber
\end{align}
where $J=\frac{dq}{dx} \in \R^3 \times \R^6$ is the top three rows of the manipulator Jacobian.

By substituting this into \eqref{eqn:QPvelocity} or \eqref{eqn:minnormexplicitsimple}, we obtain the results shown in Figure \ref{fig:v_fig}. Since this CBF does not take into account the system dynamics or the tracking ability of the low-level controller, safety is not guaranteed, but it can be achieved by proper choice of $\alpha$. In this case, with scalar multiple $\alpha \in [0.5,1]$, 
the obstacle is avoided, but not for $\alpha \in [2,3]$.

\begin{figure}[b]
    \centering
    \includegraphics[trim={1.25cm 0cm 0cm 0cm},clip,width=0.5\textwidth]{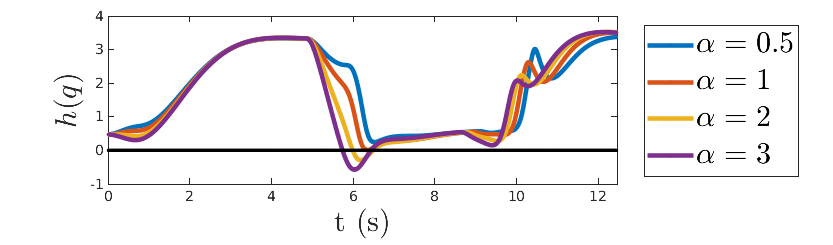}
    \caption{Velocity-based kinematic barrier function on the 6 DOF manipulator. Safety depends on choice of $\alpha$.}
    \label{fig:v_fig}
\end{figure}


\end{example}

\section{From Kinematics to Dynamics}
\label{sec:dynamics}

We now wish to establish the main result of this paper: that 
guarantees safety for the \emph{dynamics} of a robotic system. To do this, we first introduce an alternative formulation of the energy-based CBFs shown in \cite{KolathayaEnergyCBF2020} for robotic systems.

\newsec{Robotic Systems.}  We consider a robotic (or mechanical) system with configuration coordinates $q \in Q\subset \R^k$ and equations of motion:
\begin{eqnarray}
\label{eqn:roboticsys}
D(q) \ddot{q} + C(q,\dot{q}) \dot{q} + G(q) = B u
\end{eqnarray}
where $B \in \R^{k \times m}$ is the actuation matrix.  We assume $m \leq k$, wherein $m = k$ with $B$ invertible corresponds to full actuation.
From the equations of motion, we can obtain a control system of the form \eqref{eqn:controlsys}.  We will first discuss the fully actuated case, and the underactuated case will be discussed in Section \ref{sec:underactuated}.


\newsec{Energy-based Safety Constraints.}
We begin by formulating a safety-critical controllers for fully actuated robotic systems given kinematic safety constraints---thus bridging the divide from kinematic to dynamics.  This will be achieved via a ``dynamically consistent'' extension to the desired safe set. This is similar to the extensions shown in \cite{nguyen2016exponential}, \cite[Section IV]{ames2019control} for higher relative degree systems, but leverages the kinetic energy of the system.
Specifically, to dynamically extend the CBF, we note that the inertia matrix, $D(q)$ is a symmetric positive definite matrix, $D(q) = D(q)^T \succ 0$, and thus:
$$
\lambda_{\min}(D(q)) \| q \|^2 \leq q^T D(q) q \leq
\lambda_{\max}(D(q)) \| q \|^2
$$
where $\lambda_{\min}$ and $\lambda_{\max}$ are the maximum and minimum eigenvalues (which are dependent on $q$) of $D(q)$ which are necessarily positive due to the positive definite nature of $D(q)$.

\begin{definition}\label{def:energybconstraint}
Given a kinematic safety constraint expressed as a function $h : Q \subset \R^k \to \R$ only dependent on $q$, and the corresponding safe set: $\S = \{ (q,\dot{q}) \in Q \times \R^k ~ : ~ h(q) \geq 0\}$, the associated \textbf{energy-based safety constraint} is defined as:
\begin{eqnarray}
\label{eqn:dynCBF}
h_D(q, \dot{q}) := - \frac{1}{2} \dot{q}^T D(q) \dot{q} + \alpha_e h(q) \geq 0
\end{eqnarray}
with $\alpha_e > 0$.  The corresponding \textbf{energy-based safe set} is:
\begin{align}
\label{eqn:SD}
\S_D & : =   \{ (q,\dot{q}) \in Q \times \R^k ~ : ~ h_D(q,\dot{q}) \geq 0 \} .
\end{align}
\end{definition}

This construction is similar to the augmentation of kinetic energy in \cite{KolathayaEnergyCBF2020} for reciprocal control barrier functions.
While the reciprocal formulation has the advantage of having no added conservatism, due to the set remaining unchanged, it does not have well-defined behavior on the boundary of the set and outside of it, making it less popular for implementation. In fact, we now will show that the energy based constraint in Definition \ref{def:energybconstraint} is a valid (zeroing) control barrier function (CBF), thereby allowing for a new class of QPs that guarantee safety.  First, we establish the relationship between $\S_D$ and $\S$.

\begin{proposition}
Consider a kinematic safety constraint, $h: Q \subset \R^k \to \R$, with corresponding safe set $\S$, and the associated energy-based safety constraint, $h_D$, as given in Definition \ref{def:energybconstraint} with corresponding safe set $\S_D$.  Then
\begin{eqnarray}
\mathrm{(i)} ~ \S_D \subset \S, \qquad \qquad \mathrm{(ii)} ~  \mathrm{Int}(\S) \subset \lim_{\alpha_e \to \infty} \S_D  \subset \S.
\end{eqnarray}
\end{proposition}

\begin{proof}
To establish (i), we simply note that:
\begin{align}
\S_D  \subset  \{ (q,\dot{q}) \in Q \times \R^k ~ : ~ h(q)  \geq \frac{1}{2} \frac{\lambda_{\min}(D(q))}{\alpha_e} \| \dot{q}  \|^2 \}  \subset \S   \nonumber
\end{align}
which follows from the fact that: $\frac{1}{2} \lambda_{\min}(D(q)) \| \dot{q}  \|^2 \geq 0$.  To establish (ii), we first note that
\begin{align}
    S_D(\alpha_e) &= \{ (q,\dot{q}) \in Q \times \R^k ~ : ~ h_D(q, \dot{q})\geq 0\}  \nonumber\\
    &=\{ (q,\dot{q}) \in Q \times \R^k ~ : ~ h(q) \geq \frac{\frac{1}{2}\dot{q}^TD(q)\dot{q}}{\alpha_e}\} \nonumber
\end{align}
where here we made the dependence of $S_D$ on $\alpha_e$ explicit.  Consider an increasing sequence $\alpha_e^i$ where $i \in \mathbb{N}$ and $\lim_{i \to \infty}  \alpha_e^i \to \infty$.  This results is a nondecreasing sequence of sets: $\{\S_D(\alpha_e^i)\}_{i = 1}^{\infty}$:
\begin{eqnarray}
\alpha_e^i < \alpha_e^{i+1} \quad & \Rightarrow &  \quad  \frac{\frac{1}{2}\dot{q}^TD(q)\dot{q}}{\alpha_e^i}  > \frac{\frac{1}{2}\dot{q}^TD(q)\dot{q}}{\alpha_e^{i+1}} \nonumber\\
\quad & \Rightarrow &  \quad
\S_D(\alpha_e^i) \subset \S_D(\alpha_e^{i+1}). \nonumber
\end{eqnarray}
As a result:
\begin{gather*}
\lim_{i \to \infty}  \frac{\frac{1}{2}\dot{q}^TD(q)\dot{q}}{\alpha_e^i} = 0 ~ \Rightarrow ~
\lim_{i \to \infty} \S_D(\alpha_e^i) = \bigcup_{i \in \mathbb{N}} \S_D(\alpha_e^i) \supset \mathrm{Int}(\S)
\end{gather*}
and $\S_D(\alpha_e^i) \subset \S$ for all $i \in \mathbb{N}$.
\end{proof}

\newsec{Main result.}  We now have the necessary constructions to present the main result of this paper---a largely model independent safety-critical controller that ensures the forward invariance of $\S_D$ and, therefore, $\S$ in the limit for $\alpha_e$ sufficiently large.  We will establish this by showing that $h_D$ is a valid CBF and that $\dot h_D$ only depends on the kinematics, the gravity vector $G(q)$, and the inertial matrix $D(q)$.  This makes the controller more robust to uncertainty in the dynamics than full model based controllers---which would require knowledge of the Coriolis matrix, $C(q,\dot{q})$.

\begin{theorem} \label{thm:fully_actuated}
  Consider a robotic system \eqref{eqn:roboticsys}, assumed to be fully actuated with $B$ invertible, and a kinematic safety constraint $h : Q \to \R$ with corresponding safe set $\S = \{(q,\dot{q}) \in Q \times \R^k  :  h(q) \geq 0\}$.  Let $h_D$ be the energy based constraint defined as in \eqref{eqn:dynCBF} with corresponding safe set $\S_D$ as given in \eqref{eqn:SD}.
  Then $h_D$ is a control barrier function on $\S_D$ and given a desired controller $u_{\rm des}(x,t)$, the following controller for all $(q,\dot{q}) \in \S_D$:
\begin{gather}
\label{eqn:QProbotic}
\begin{split}
u^*(q,\dot{q},t) =  ~ \underset{u \in \R^m}{\operatorname{argmin}}  ~ & ~ \| u - u_{\rm des}(q,\dot{q},t) \|^2  \\
~  \mathrm{s.t.}   ~ &  ~\underbrace{- \dot{q}^T B u + G(q)^T \dot{q} + \alpha_e J_h(q)\dot{q}  }_{\dot{h}_D(q,\dot{q},u)} \geq - \alpha(h_D(q,\dot{q})),
\end{split}
\end{gather}
guarantees forward invariance of $\S_D$, i.e., safety of $\S_D$.  Additionally, it has a closed form solution:
\begin{gather}
\begin{split}
\label{eqn:minnormexplicitrobotic}
  u^*(x,t)  = u_{\rm des}(q,\dot{q},t) +
\left\{
\begin{array}{lcr}
\frac{B^T \dot{q}}{\|B^T  \dot{q} \|^2} \Psi(x,t;u_{\rm des})  & \mathrm{if~} \Psi(x,t; u_{\rm des}) < 0 \\
0 & \mathrm{if~} \Psi(x,t; u_{\rm des})  \geq 0
\end{array}
\right.
\end{split}
\end{gather}
where
\begin{gather}
\Psi(x,t;u_{\rm des}) := \dot{q}^T (\alpha_e J_h(q)^T  +  G(q) - B u_{\rm des}(x,t) ) + \alpha(h_D(q,\dot{q})) .  \nonumber
\end{gather}
\end{theorem}
It is interesting to note that $h_D$ is a CBF on $\S_D$ without requiring that $h$ has relative degree 1, i.e., one need not require that $J_h(q) \neq 0$ (except on $\partial \S$) as in Lemma \ref{lem:velocity}. This reinforces the idea that these energy-based control barrier functions are a natural extension for relative-degree 2 robotic systems.



\begin{proof}[Proof of Theorem \ref{thm:fully_actuated}]
Differentiating $h_D$ along solutions yields (and suppressing the dependence on $q$ and $\dot{q}$):
\begin{eqnarray}
\label{eqn:hDcalc}
\dot{h}_D  & = &   -\dot{q}^T D  \ddot{q} -\frac{1}{2} \dot{q}^T \dot{D} \dot{q} + \alpha_e J_h \dot{q} \\
& = &  \dot{q}^T \left( C  \dot{q} + G  -  B u \right)  - \frac{1}{2} \dot{q}^T \dot{D} \dot{q} + \alpha_e J_h  \dot{q}  \nonumber\\
& = &  \frac{1}{2} \dot{q}^T \left( - \dot{D} + 2 C \right)\dot{q}  - \dot{q}^T  B u  + G^T \dot{q} + \alpha_e J_h \dot{q} \nonumber\\
& =  &  - \dot{q}^T  B u  + G^T \dot{q} + \alpha_e J_h  \dot{q}  \nonumber
\end{eqnarray}
where the last 
equality follows from the fact that $\dot{D} - 2 C$ is skew symmetric. To establish that $h_D$ is a CBF, we need only show that \eqref{eqn:QProbotic} has a solution since the inequality constraint in \eqref{eqn:QProbotic} implies that \eqref{eqn:cbf:definition} is satisfied in Definition \ref{def:cbf}.
As a result of Lemma \ref{lem:explicitCBFsol}, the solution to \eqref{eqn:QProbotic} is given by \eqref{eqn:minnormexplicit}; this follows by noting:
$$
L_f h_D(q,\dot{q})= (\alpha_e J_h(q) + G(q)^T ) \dot{q},
\quad
L_g h_D(q,\dot{q}) = -\dot{q}^T B.
$$
To show that \eqref{eqn:minnormexplicit} is well defined, we need to establish that:
\begin{gather}
L_g h_D(q,\dot{q}) = -\dot{q}^T B = 0  \quad \Rightarrow \quad L_f h_D (q,\dot{q}) + \alpha(h_D(q,\dot{q})) \geq 0. \nonumber
\end{gather}
Yet $\dot{q}^T B = 0$ implies that $\dot{q}^T= 0$ since $B$ is invertible and therefore
$L_f h_D(q, \dot{q}) = 0$ and since $(q,\dot{q}) \in \S_D$ it follows that $h_D(q,\dot{q}) \geq 0$ and hence $\alpha(h_D(q,\dot{q})) \geq 0$ implying that \eqref{eqn:minnormexplicit} is well defined and thus $h_D$ is a CBF.  Finally, the results Lipschitz continuity and forward invariance of $\S_D$ follow from the results of Lemma \ref{lem:explicitCBFsol} and Theorem \ref{thm:cbf}.
\end{proof}

Having established Theorem \ref{thm:fully_actuated}, the following corollary demonstrates how to further reduce model dependence.

\begin{corollary}
\label{rem:model-free}
Under the conditions of Theorem \ref{thm:fully_actuated}, if there exists a $c_u > 0$ such that $c_u \geq \frac{1}{2}\lambda_{\max}(D(q))$ then replacing the safety constraint \eqref{eqn:QProbotic} in the safety-critical QP with:
\begin{equation}\label{eqn:firstineqhdDfree}
  ~\underbrace{- \dot{q}^T B u + G(q)^T \dot{q} + \alpha_e J_h(q)\dot{q}  }_{\dot{h}_D(q,\dot{q},u)} \geq - \alpha(-
c_u\norm{\dot{q}}^2 + \alpha_e h(\q)),
\end{equation}
implies safety of $\S_D$.  Moreover, if in addition $\|G(\q)\|\leq c_u$, for $c_u>0$, then the constraint \eqref{eqn:QProbotic} can be replaced by:
\begin{equation}\label{eqn:secondineqhdDfree}
  \alpha_e J_h(q)\dot{q}  - \dot{q}^T B u - c_u |\dq|   \geq - \alpha\left(-c_u\norm{\dot{q}}^2 + \alpha_e h(\q)\right).
\end{equation}
wherein safety of $\S_D$ is guaranteed.
\end{corollary}

\begin{proof}
It can be verified that $-\alpha(-c_u\norm{\dot{q}}^2 + \alpha_e h(\q)) \geq -\alpha(-\frac{1}{2}
\lambda_{\max}(D(q))\norm{\dot{q}}^2 + \alpha_e h(\q)) \geq -\alpha(h_D(q,\dot{q}))$, which means that \eqref{eqn:firstineqhdDfree}$\implies$\eqref{eqn:QProbotic}.  The second inequality, \eqref{eqn:secondineqhdDfree}, follows from the bound on the gravity vector $G$.
\end{proof}

\begin{figure*}[t]
    \centering
    \includegraphics[trim={.4cm .5cm .5cm .2cm},clip,width=\textwidth]{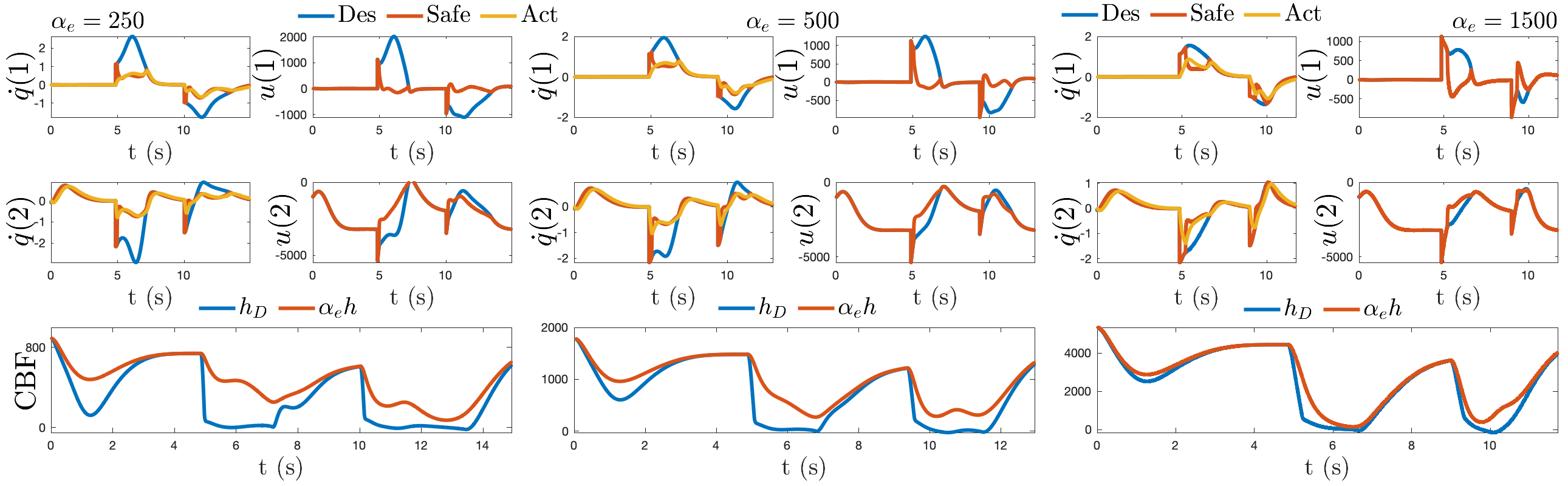}
    \caption{Energy-based kinematic CBF on the 6 DOF manipulator. Safety is guaranteed regardless of the choice of $\alpha_e$, but performance improves as $\alpha_e$ increases.}
    \label{fig:D_fig}
\end{figure*}

\ifdefined\revtwo
\else
\textcolor{red}{If the $D$ is not known, or is uncertain, the safety guarantees from Theorem \ref{thm:fully_actuated} hold under the alternative constraint
\begin{equation}
    ~\underbrace{\alpha_e J_h(q)\dot{q}  - \dot{q}^T B u + G(q)^T \dot{q}}_{\dot{h}_D(q,\dot{q},u)} \geq - \alpha(\tilde{h}(q,\dot{q})),
\end{equation}
with $\tilde{h}(q,\dot{q}) \leq h_D(q,\dot{q})\ \forall (q,\dot{q}) \in \mathcal{S}_D$. For example, if $\tilde{h} = -\frac{1}{2}
\lambda_{\min}(D(q))\norm{\dot{q}}^2 + \alpha_e h$, then $\tilde{h} \leq h_D$ and safety is guaranteed, since:
\begin{equation}
    ~\underbrace{\alpha_e J_h(q)\dot{q}  - \dot{q}^T B u + G(q)^T \dot{q}}_{\dot{h}_D(q,\dot{q},u)} \geq - \alpha(\tilde{h}(q,\dot{q})) \geq -\alpha(h_d(q,\dot{q})).
\end{equation}
This leads to the following remark for synthesizing a controller that is robust to model uncertainty without adding conservatism.}\fi

\ifdefined\revtwo
\else
\begin{remark}\textcolor{red}{
Consider the controller $u = B^{-1} G(q) +  \mu$, The QP in \eqref{eqn:QProbotic} becomes:
\begin{gather}
\label{eqn:QProbotic}
\begin{split}
u^*(q,\dot{q},t) =  ~ \underset{\mu \in \R^m}{\operatorname{argmin}}  ~ & ~ \| B^{-1} G(q) + \mu - u_{\rm des}(q,\dot{q},t) \|^2  \\
~  \mathrm{s.t.}   ~ & ~\underbrace{\alpha_e J_h(q)\dot{q}  - \dot{q}^T B \mu }_{\dot{h}_D(q,\dot{q},\mu)} \geq - \alpha\left(-\frac{1}{2}
\lambda_{\min}(D(q))\norm{\dot{q}}^2 + \alpha_e h(q)\right),
\end{split} \nonumber
\end{gather}
As $\alpha_e \rightarrow \infty$, we have that $\alpha_e h \rightarrow h_D$ and the QP becomes:
\begin{gather}
\label{eqn:QProbotic}
\begin{split}
u^*(q,\dot{q},t) =  ~ \underset{\mu \in \R^m}{\operatorname{argmin}}  ~ & ~ \| B^{-1} G(q) + \mu - u_{\rm des}(q,\dot{q},t) \|^2  \\
~  \mathrm{s.t.}   ~ & ~\underbrace{\alpha_e J_h(q)\dot{q}  - \dot{q}^T B \mu }_{\dot{h}_D(q,\dot{q},\mu)} \geq - \alpha\left(\alpha_e h(q)\right), \end{split} \nonumber
\end{gather}
This QP is especially robust because the inequality constraint does not depend on any model information, i.e., it is \emph{model independent}.  Note that this QP does not formally guarantee safety but practically does so as $\alpha_e \to \infty$. }
\end{remark}
\fi

\newsec{Connections with kinematic control.}  The goal is to now connect the previous constructions with the kinematic controllers defined in Section \ref{sec:kinematic}. Often, controllers can only be implemented as desired position and velocity commands that are passed to embedded level PD controllers. As such, we consider a controller of the form:
\begin{eqnarray}
\label{eqn:uDcont}
u = - K_{\rm vel}(\dot{q} - \dot{q}^*_d(q,\dot{q},t))
\end{eqnarray}
where $\dot{q}^*_d(q,t)$ is a desired velocity signal that enforces safety while trying to achieve tracking as in the case of Lemma \ref{lem:velocity} wherein we have a desired velocity based tracking controller: $\dot{q}_{\rm des}(q,t)  := J_y(q)^{\dagger} \left(\dot{x}_d(t) - \lambda (y(q) - x_d(t)) \right)$ for $\lambda > 0$.  The following is a result of the direct application of Theorem \ref{thm:fully_actuated} in the context of the controller \eqref{eqn:uDcont}.

\begin{theorem} \label{thm:fully_actuated-kinematiccontrol}
 Consider a robotic system \eqref{eqn:roboticsys}, and assume it is fully actuated.  Given a kinematic safety constraint $h : Q \to \R$ and the associated dynamically consistent extended CBF $h_D : Q \times \R \to \R$ as given in \eqref{eqn:dynCBF} with associated safe set $\S_D$, along with a desired trajectory $x_d(t)$ in the task space $x = y(q)$.  The D controller \eqref{eqn:uDcont} with $K_{\rm vel} \succ 0$ and the following QP:
\begin{gather}
\label{eqn:QProboticD}
\begin{split}
\dot{q}^*_d= ~
&  \underset{\dot{q}_d \in \R^n}{\operatorname{argmin}}  ~  ~ \| \dot{q}_d - \overbrace{J_y^{\dagger} \left(\dot{x}_d - \lambda (y - x_d) \right)}^{\dot{q}_{\rm des}(q,t)} \|^2  \\
&     \mathrm{s.t.}   ~  ~\underbrace{\alpha_e J_h \dot{q} +  \dot{q}^T B K_{\rm vel} \dot{q}  - \dot{q}^T B K_{\rm vel} \dot{q}_d + G^T \dot{q}}_{\dot{h}_D(q,\dot{q},\dot{q}_d)} \geq - \alpha(h_D),
\end{split}
\end{gather}
guarantees forward invariance, i.e., safety, of $\S_D$. Moreover, it has a closed form solution:
\begin{gather}
\begin{split}
\label{eqn:minnormexplicitrobotickvel}
  \dot{q}^*_d= \dot{q}_{\rm des} +
\left\{
\begin{array}{lcr}
\frac{K_{\rm vel}^T B^T \dot{q}}{\|K_{\rm vel}^T B^T  \dot{q} \|^2} \Psi(q,\dot{q},t;q_{\rm des})  & \mathrm{if~} \Psi(q,\dot{q},t;q_{\rm des}) < 0 \\
0 & \mathrm{if~} \Psi(q,\dot{q},t;q_{\rm des}) \geq 0
\end{array}
\right.
\end{split}
\end{gather}
where
\begin{gather}
\Psi(q,\dot{q},t;
\dot{q}_{\rm des}):= \dot{q}^T ( \alpha_e J_h^T   +
B K_{\rm vel} \dot{q} - B K_{\rm vel} \dot{q}_{\rm des}+  G ) + \alpha(h_D) .  \nonumber
\end{gather}
\end{theorem}
Proof of Theorem \ref{thm:fully_actuated-kinematiccontrol} is omitted as it is a straightforward extension of Theorem \ref{thm:fully_actuated}.
\ifdefined\revtwo
\else
\begin{corollary}
\textcolor{red}{Consider a robotic system \eqref{eqn:roboticsys}, and assume it is fully actuated.  Given a kinematic safety constraint $h : Q \to \R$ and the associated dynamically consistent extended CBF $h_D : Q \times \R \to \R$ as given in \eqref{eqn:dynCBF} with associated safe set $\S_D$, along with a desired trajectory $x_d(t)$ in the task space $x = y(q)$.  The D controller \eqref{eqn:uDcont} with $K_{\rm vel} \succ 0$ and:
\begin{gather}
\label{eqn:QProboticD}
\begin{split}
\dot{q}^*_d= ~
&  \underset{\dot{q}_d \in \R^n}{\operatorname{argmin}}  ~  ~ \| \dot{q}_d - \overbrace{J_y^{\dagger} \left(\dot{x}_d - \lambda (y - x_d) \right)}^{\dot{q}_{\rm des}(q,t)} \|^2  \\
&     \mathrm{s.t.}   ~  ~\underbrace{\alpha_e J_h \dot{q} +  \dot{q}^T B K_{\rm vel} \dot{q}  - \dot{q}^T B K_{\rm vel} \dot{q}_d + G^T \dot{q}}_{\dot{h}_D(q,\dot{q},\dot{q}_d)} \geq - \alpha(h_D),
\end{split}
\end{gather}
guarantees forward invariance, i.e., safety, of $\S_D$.  Moreover, it has a closed form solution:
\begin{gather}
\begin{split}
\label{eqn:minnormexplicitrobotic}
  \dot{q}^*_d= \dot{q}_{\rm des} +
\left\{
\begin{array}{lcr}
\frac{K_{\rm vel}^T B^T \dot{q}}{\|K_{\rm vel}^T B^T  \dot{q} \|^2} \Psi(q,\dot{q},t;q_{\rm des})  & \mathrm{if~} \Psi(q,\dot{q},t;q_{\rm des}) < 0 \\
0 & \mathrm{if~} \Psi(q,\dot{q},t;q_{\rm des}) \geq 0
\end{array}
\right.
\end{split}
\end{gather}
where
\begin{gather}
\Psi(q,\dot{q},t;
\dot{q}_{\rm des}):= \dot{q}^T ( \alpha_e J_h^T  +  G +
B K_{\rm vel} \dot{q} - B K_{\rm vel} \dot{q}_{\rm des} ) + \alpha(h_D) .  \nonumber
\end{gather}}
\end{corollary}\fi
It may be the case, as with industrial actuators, that $K_{\rm vel}$ is not known.  In that case, it can typically be determined from experimental data.  Formally, one can guarantee safety by utilizing adaptive control barrier functions \cite{taylor_adaptive_barrier}.
Similar to Remark \ref{rem:model-free}, we can reformulate the constraints to eliminate the $D$ and $G$ matrices to yield robust QPs. 
\ifdefined\revtwo
\else
\alert{
\begin{remark}
In the case when one has the ability to achieve torque control on the robot, one could consider the following modification of \eqref{eqn:uDcont}:
$$
u = B^{-1} G(q) - K_{\rm vel}(\dot{q} - \dot{q}^*_d(q,\dot{q},t))
$$
This can be merged with the observations in Remark \ref{rem:modelfree} to obtain a QP of the form:
\begin{gather}
\label{eqn:QProboticDmodelfree}
\begin{split}
\dot{q}^*_d= ~
&  \underset{\dot{q}_d \in \R^n}{\operatorname{argmin}}  ~  ~ \| \dot{q}_d - \overbrace{J_y^{\dagger} \left(\dot{x}_d - \lambda (y - x_d) \right)}^{\dot{q}_{\rm des}(q,t)} \|^2  \\
&     \mathrm{s.t.}   ~  ~\underbrace{\alpha_e J_h \dot{q} +  \dot{q}^T B K_{\rm vel} \dot{q}  - \dot{q}^T B K_{\rm vel} \dot{q}_d }_{\dot{h}_D(q,\dot{q},\dot{q}_d)} \geq - \alpha(\alpha_e h),
\end{split}
\end{gather}
that enforces safety in the limit as $\alpha_e \to \infty$.  Important, this QP is not totally model free in that it does not directly depend on any model information and, as a result, if very robust to model uncertainty.  The only model information needed is in the feed forward gravity vector, $G(q)$, and this could be replaced by a feed-forward torque term determined experimentally, e.g., via learned impedance control parameters.
\end{remark}
}\fi

\begin{example}[\textbf{Energy-based kinematic CBF}]
The 6 DOF manipulator from Example \ref{ex:manipulator_kinematic_only} is now filtered with the constraint given in \eqref{eqn:firstineqhdDfree}, using $c_u = 5 \lambda_{\max} (D)$. Figure \ref{fig:D_fig} shows the result with a variety of different $\alpha_e$ values. Safety is guaranteed regardless of the value of $\alpha_e$, but as the value increases, the manipulator is able to move faster and get closer to obstacles, resulting in a higher performing system with the same safety guarantees.



\end{example}

\ifdefined\revtwo
\else
\newsec{Application to underatuated systems.}  The methods developed can also be applied to underactuated systems, i.e., where $m \leq k$ and we have a potentially non-singular actuation matrix $B$. .  In this context, we take inspiration from the classical operational space formulation \cite{khatib1987unified}.  Specifically, by applying the ``pre-feedback'' controller $u = B^T J_h^T u_h$, for $u_h \in \R$ the \emph{safety input}.  This yields:

\begin{corollary}
Consider a robotic system \eqref{eqn:roboticsys} and a kinematic safety constraint: $h : Q \to \R$ together with the dynamically extended CBF $h_D$ and associated safe sets $\S$ and $\S_D$, respectively.  If for all $(q,\dot{q}) \in \S_D$:
\begin{align}
\label{eqn:CBFundercond}
\dot{q}^T B B^T J_h(q)^T  = 0 ~  \Rightarrow  ~
(G^T  + \alpha_e J_h )  \dot{q} \geq - \alpha(h_D(q,\dot{q}))
\end{align}
then the following controller:
\begin{gather}
\label{eqn:QProboticunder}
 u(q,\dot{q},t) =   B^T J_h(q)^T u_h^* (q,\dot{q},t) \hspace{0.6\columnwidth} \\
\begin{split}
u^*(q,\dot{q},t) = & ~ \underset{u_h \in \R}{\operatorname{argmin}}  ~  ~ \| B^T J_h^T u_h - u_{\rm des}(q,\dot{q},t) \|^2  \\
& \quad   \mathrm{s.t.}   ~  ~\underbrace{\alpha_e J_h(q)\dot{q}  - \dot{q}^T B B^T J_h(q)^T u_h  + G(q)^T \dot{q}}_{\dot{h}_D(q,\dot{q},u_h)} \geq - \alpha(h_D(q,\dot{q})),
\end{split} \nonumber
\end{gather}
guarantees forward invariance of $\S_D$, i.e., safety of $\S_D$, and has a closed form solution.
\end{corollary}

\begin{proof}
Following from the calculation in \eqref{eqn:hDcalc}:
$$
\dot{h}_D = \underbrace{- \dot{q}^T  B B^T J_h(q)^T}_{L_g h_D(q, \dot{q})} u_h  + \underbrace{(G^T  + \alpha_e J_h)  \dot{q}}_{L_f h_D(q,\dot{q})} ,
$$
wherein the condition in \eqref{eqn:CBFundercond} implies that $h_D$ is a CBF.
\end{proof}


\alert{Alternative formulation---merging the two extended CBF formulations.}

\newsec{Application to underatuated systems.}
We now consider the application of this methodology to underactuated systems, i.e., where $m \leq k$ and we have a potentially non-singular actuation matrix $B$.  In this setting, we look to merge the methods presented in Section \ref{sec:secondorder}, and specifically Lemma \ref{lem:extendedcbf}, with the dynamically consistent extended CBF.  In this context, consider:
\begin{eqnarray}
\label{eqn:heunder}
h_{u}(q, \dot{q}) := -(\dot{h}(q,\dot{q}))^2 + \alpha_e h(q)
\end{eqnarray}
As in the case of $h_D$, we have that for $\S_u = \{(q,\dot{q}) \in Q \times \R ~ : ~ h_u(q,\dot{q}) \geq 0$ that $\S_u \subset \S$ where $\lim_{\alpha_e \to \infty} \S_u \to \S$.

\subsection{Safety-Critical Control of Underactuated Robotic Systems}

We now consider the application of this methodology to underactuated systems, i.e., where $m \leq k$ and we have a potentially non-singular actuation matrix $B$.  In this setting, we look to generalize the methods presented in Section \ref{sec:secondorder}, and specifically Lemma \ref{lem:extendedcbf}.  

\newsec{Safety-critical operation space.}  Taking inspiration from the classical operational space formulation \cite{khatib1987unified} of converting the system to the ``task dynamics.''  Specifically, by applying the ``pre-feedback'' controller $u = B^T J_h^T u_h$, for $u_h \in \R$ the \emph{safety input}.  Premultiplying by the transpose of the \emph{dynamically consistent} pseudoinverse:
$$
\widehat{J}_h  =  \frac{D^{-1} J_h^T}{J_h D^{-1} J_h^T} \quad
\Rightarrow  \quad
J_h \widehat{J}_h  = I
$$
yields the \underline{CBF dynamics}:
\begin{gather}
\underbrace{\frac{1}{ J_h D^{-1} J_h^T }}_{:= D_h(q)} \ddot{h} +
\underbrace{\widehat{J}_h^T \left( C \dot{q} + G - D \dot{J}_h \dot{q} \right)}_{:= H_h(q,\dot{q})} =
\underbrace{\widehat{J}_h^T B B^T J_h^T}_{:= B_h(q)} u_h
\end{gather}
which describes the evolution of the control barrier function, $h$, as a mechanical system. This motivates the following result.

\begin{theorem}
Consider a robotic system \eqref{eqn:roboticsys} and a kinematic safety constraint: $h : Q \to \R$.  Consider the extended CBF for underactuated systems, for $\alpha_e > 0$:
\begin{eqnarray}
\label{eqn:heunderactated}
h_{e}(q, \dot{q}) := \dot{h}(q,\dot{q}) + \alpha_e h(q)
\end{eqnarray}
with safe set: $\S_e := \{ (q, \dot{q} \in Q \times \R ~ : ~ h_e(q,\dot{q}) \geq 0\}$.  Then if $B_h(q) \neq 0$, the following controller:
\begin{gather}
\label{eqn:QProbotic}
\begin{split}
 u =  & B^T J_h^T u_h^* \\
u^* = & ~ \underset{u_h \in \R}{\operatorname{argmin}}  ~  ~ \| B^T J_h^T u_h - u_{\rm des} \|^2  \\
& \quad   \mathrm{s.t.}   ~  ~\underbrace{J_h(q)\dot{q}  - \dot{q}^T B u + G(q) \widehat{J}_h \dot{q}}_{\dot{h}_e(q,\dot{q},u_h)} \geq - \alpha(h_e),
\end{split}
\end{gather}
guarantees forward invariance of $\S_e \cap \S$, i.e., safety of $\S_e \cap \S$, stability of $\S_e$, and has a closed form solution.

\end{theorem}

\begin{proof}
We begin by deriving the CBF dynamics and, thereby, verifying that the barrier function, $h$, evolves according to these dynamics (here we follow the notation from \cite{ames2013towards}).  Beginning with the robot dynamics \eqref{eqn:roboticsys}, premultiplying by $\widehat{J}_h^T$ and applying the controller $u = B^T J_h^T u_h$ yields:
$$
\widehat{J}_h^T D \ddot{q} + \widehat{J}_h^T C(q,\dot{q}) \dot{q})
+ \widehat{J}_h^T G(q) = \widehat{J}_h^T B B^T J_h^T u_h
$$
Now $\dot{h} = J_h(q) \dot{q}$ and thus $\ddot{h} = J_h(q) \ddot{q} + \dot{J}_h(q,\dot{q}) \dot{q}$.  Premultiplying by $D_h$ as defined in \eqref{eqn:dynCBF} yields:
\begin{eqnarray}
D_h \ddot{h} & = &  D_h J_h \ddot{q} + D_h \dot{J}_h \dot{q} \nonumber\\
& = & \widehat{J}_h^T D \ddot{q} + D_h J_h \dot{q}  \nonumber\\
& = & \underbrace{- \widehat{J}_h^T H + D_h J_h \dot{q} }_{H_h} + \underbrace{\widehat{J}_h^T B B^T J_h^T}_{B_h} u_h \nonumber
\end{eqnarray}
where here we used the fact that:
$$
\widehat{J}_h^T D =
\frac{J_h D^{-1}}{J_h D^{-1} J_h^T}  D
= D_h J_h
$$
and therefore the end result are the CBF dynamics in \eqref{eqn:dynCBF}.
\end{proof}

\begin{theorem}
Consider a robotic system \eqref{eqn:roboticsys} and a kinematic safety constraint: $h : Q \to \R$.  Consider the dynamically consistent extended CBF for underactuated systems:
\begin{eqnarray}
\label{eqn:dynCBFunder}
\widehat{h}_{D}(q, \dot{q}) := - \frac{1}{2} \dot{h}(q,\dot{q})^T D_h(q) \dot{h}(q,\dot{q}) + h(q)
\end{eqnarray}
with safe set: $\widehat{S}_D := \{ (q, \dot{q} \in Q \times \R ~ : ~ \widehat{h}_D(q,\dot{q}) \geq 0\}$.  Then $\widehat{S}_D \subset \S$ and for all $(q,\dot{q}) \in \widehat{S}_D$ the controller:
\begin{gather}
\label{eqn:QProbotic}
\begin{split}
u^*(q,\dot{q},t) =  ~ \underset{u_h \in \R}{\operatorname{argmin}}  ~ & ~ \| B^T J_h^T u_h - u_{\rm des}(q,\dot{q},t) \|^2  \\
~  \mathrm{s.t.}   ~ & ~\underbrace{J_h(q)\dot{q}  - \dot{q}^T B u + G(q) \widehat{J}_h \dot{q}}_{\dot{h}_D(q,\dot{q},u)} \geq - \alpha(h_D(q,\dot{q})),
\end{split}
\end{gather}
guarantees forward invariance of $\widehat{S}_D$, i.e., safety of $\widehat{S}_D$, and has a closed form solution:

\end{theorem}

\fi

\newtheorem{property}{\bfseries Property}

\section{Underactuated systems}
\label{sec:underactuated}
    The methods developed can also be applied to underactuated systems, i.e., where $m \leq k$ and we have a potentially non-singular actuation matrix $B$. The key idea is to treat $h(\q)$ as one of the coordinates. Therefore, assume a mapping $\Phi(\q) := (w(\q),h(\q))$, where $w$ is chosen such that $\Phi$ is a diffeomorphism. We obtain the derivative as
\begin{align}
\begin{bmatrix}
\dot w(\q,\dq) \\ \dot h(\q,\dq) 
\end{bmatrix} = J_e(\q) \dq,
\end{align}
where $J_e(\q)$ is the Jacobian matrix. $J_e$ is non-singular by property of diffeomorphism.
We re-write the equations of motion of the robot as
\begin{align}\label{eq:newdynamics}
D_e(\q) \begin{bmatrix}
\ddot w \\ \ddot h
\end{bmatrix} + C_e(\q , \dot q) \begin{bmatrix}
\dot w \\ \dot h
\end{bmatrix} + G_e(\q)  =  J_e(\q)^{-T} B u, 
\end{align}
where 
\begin{align}
D_e(\q)  &= J_e(\q)^{-T} D(\q) J_e(\q)^{-1} \nonumber \\
C_e(\q,\dq) &= 	J_e(\q)^{-T}  C(\q) J_e(\q)^{-1} + J_e(\q)^{-T} D(\q) \dot {J}_e(\q)^{-1}  \nonumber \\
G_e(\q) &= J_e(\q)^{-T} G(\q),
\end{align}
are the new terms that define the dynamics in the transformed space. It can be verified that the properties of $D_e, C_e$ will be same as that of $D, C$, i.e., $D_e$ is symmetric positive definite, and $\dot D_e - 2 C_e$ is skew-symmetric. The interested reader may see \cite[Chapter 4, Section 5.4]{murray1994mathematical}
for details on the coordinate change and the corresponding
properties. We separate the robotic dynamics \eqref{eq:newdynamics} into two parts:
\begin{align}\label{eq:underactuatedeom}
D_{11}(\q) \ddot w + D_{12}(\q) \ddot h + C_1(\q,\dq) \dq + G_1(\q) &= B_1(\q) u \nonumber \\
D_{21}(\q) \ddot w + D_{22}(\q) \ddot h + C_2(\q,\dq) \dq + G_2(\q) &= B_2(\q) u,
\end{align}
where the terms corresponding to $D,C,G,B$ are apparent from the setup. $\ddot w$ can be eliminated from \eqref{eq:underactuatedeom} to obtain 
\begin{align}\label{eq:eliminateddynamicsunactuated}
 &\underbrace{(D_{22}  - D_{21} D_{11}^{-1} D_{12} )}_{D_h}  \ddot h +   \underbrace{(C_2 - D_{21} D_{11}^{-1} C_1 )}_{C_h} \dq  \\
 &\hspace{20mm} + \underbrace{G_2 - D_{21} D_{11}^{-1} G_1}_{G_h}  = \underbrace{(B_2 - D_{21} D_{11}^{-1} B_1)}_{B_h} u, \nonumber
\end{align}
where $D_h$ is nothing but the Schur complement form, and it is known to be symmetric positive definite \cite[Proposition 1]{kolathaya2020aut}. Note that here $B_h: Q \to \mathbb{R}^{1\times m} $ is the mapping from $u$ to the joints, which is assumed to have full row rank (in other words, $h$ is assumed to be inertially coupled with $u$. This may not be satisfied for all $Q$, in which case a subset $Q_u \subset Q$ is chosen (for example, in the cart-pole, pole-angle is not inertially coupled with $u$ when it is horizontal). With this formulation, we have the following theorem.
\begin{theorem}
Consider a robotic system \eqref{eqn:roboticsys} and a kinematic safety constraint: $h : Q \to \R$.  Consider the dynamically consistent extended CBF for underactuated systems: 
\begin{eqnarray}
\label{eqn:dynCBFunder}
\widehat{h}_{D}(q, \dot{q}) := - \frac{1}{2} \dot{h}(q,\dot{q})^T D_h(q) \dot{h}(q,\dot{q}) +\alpha_e h(q)
\end{eqnarray}
with safe set: $\widehat{S}_D := \{ (q, \dot{q}) \in Q \times \R^k ~ : ~ \widehat{h}_D(q,\dot{q}) \geq 0\}$.  Then $\widehat{S}_D \subset \S$ and for all $(q,\dot{q}) \in \widehat{S}_D$ the following controller: 
\begin{gather}
\label{eqn:QProbotic_under}
\begin{split}
u^*(q,\dot{q},t) = & ~ \underset{u \in \R^m}{\operatorname{argmin}}  ~  ~ \| u - u_{\rm des}(q,\dot{q},t) \|^2  \\
~  \mathrm{s.t.}   ~ & - \frac{1}{2} \dot h \dot D_h \dot h - \dot h (- C_h \dq - G_h)  +  \alpha_e \dot h - \dot h B_h u \geq - \alpha(\widehat h_D(\q,\dq))
\end{split}
\end{gather}
guarantees forward invariance of $\widehat{S}_D$, i.e., safety of $\widehat{S}_D$. 
\end{theorem}

\begin{proof}
Differentiating $\widehat h$ yields: 
\begin{align}
    \dot {\widehat h}_D = - \frac{1}{2} \dot h \dot D_h \dot h - \dot h (- C_h \dq - G_h) +  \alpha_e \dot h - \dot h B_h u .
\end{align}
It can be verified that if $\dot h=0$, then the inequality in \eqref{eqn:QProbotic_under} is satisfied. The safety property follows directly.
\end{proof}

\begin{remark}
Similar to Remark \ref{rem:model-free}, we can eliminate some of the model-based terms in \eqref{eqn:QProbotic_under}. Specifically, we can replace the constraint in the QP with the following constraint:
\begin{gather}
    - \frac{1}{2} c_l \dot h^2 - c_u |\dot h| (|\dq|^2+1)    + \alpha_e   \dot h - \dot h B_h  u \geq - \alpha(- c_u \dot h^2 + \alpha_e h(x)), \nonumber
\end{gather}
where $c_l,c_u$ are constants that bound the norms:
    $c_l \leq \|D_h\| \leq c_u$,  $\|C_h\|\leq c_u |\dq|$, $\|G_h\|\leq c_u$.
We have used the same notations (see Remark \ref{rem:model-free}) for convenience.
Note that these bounds may not exist for all $(\q,\dq)\in Q\times \R^k$, and they are dependent on the validity of the coordinate transformation $\Phi$. More details on the bounds are in \cite{kolathaya2020aut}.
\end{remark}

\begin{example}[\textbf{Cart-Pole System}]
To demonstrate these concepts, we consider the cart-pole system with two states, the cart position $x$ and the pole angle $\theta$. The system is actuated through a force input $u$ applied to the cart, which moves freely in a line. 
The safety constraint is to ensure that pole remains mostly upright, with $\theta \in [\frac{5\pi}{6}, \frac{7\pi}{6}]$. 
We choose $w = x$ and $h = (\frac{\pi}{6})^2-(\theta-\pi)^2$.
The results of the applying the QP \eqref{eqn:QProbotic_under} is shown in Figure  \ref{fig:cart_pole}.
\end{example}

\begin{figure}
    \centering
    \includegraphics[trim={0cm 2cm 0cm .5cm}, clip, width=.45\textwidth]{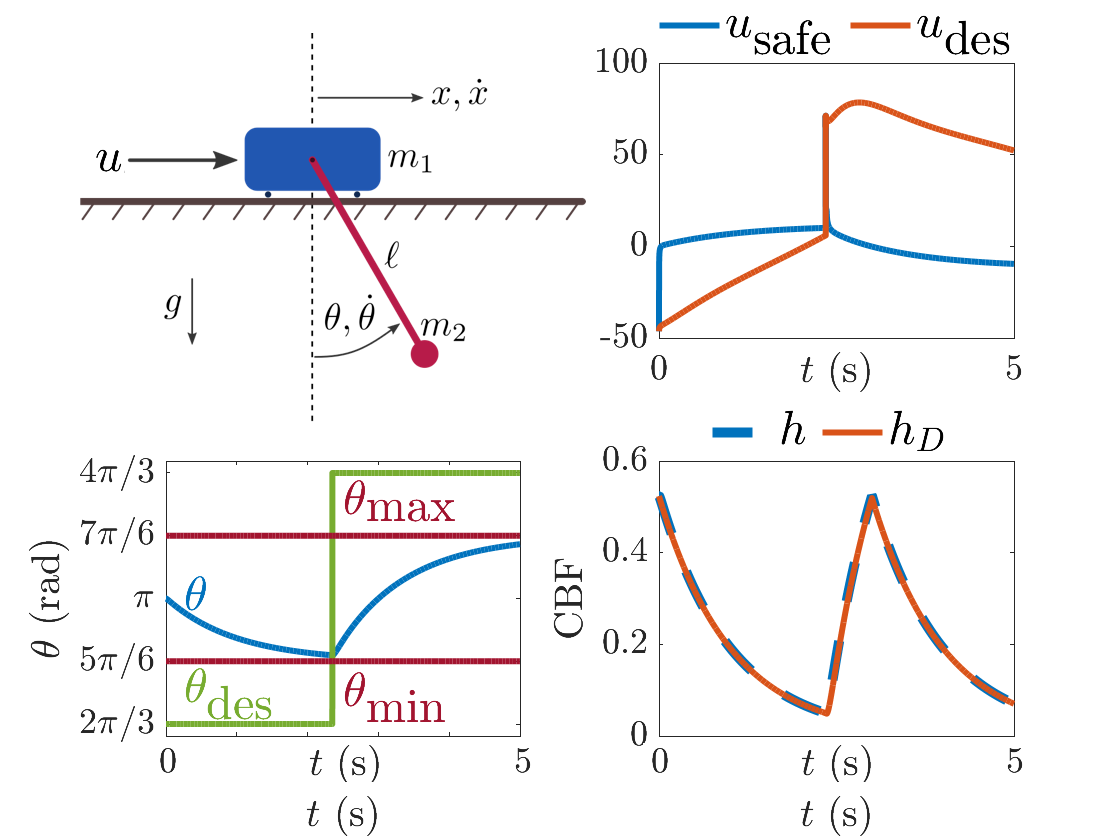}
    \caption{Cart-pole system with energy-based CBF.}
    \label{fig:cart_pole}
    \vspace{-0.5cm}
\end{figure}

\bibliographystyle{IEEEtran}
\bibliography{bib/Paper,bib/barrier,bib/kinematic}

\end{document}